\documentclass{article}
\usepackage[utf8]{inputenc}
\usepackage{amsmath}
\usepackage{amsfonts}
\usepackage{xcolor}
\usepackage{subcaption}
\usepackage{float}
\usepackage[top=1in, bottom=1in, left=1in, right=1in]{geometry}
\usepackage{amsthm}
\usepackage{mathtools}
\usepackage{graphicx}
\usepackage{indentfirst}
\usepackage{url}
\usepackage{algorithm}
\usepackage{algpseudocode}
\usepackage{bbm}
\usepackage{empheq}
\title{Upper bounds on the Rate of Uniformly-Random Codes for the Deletion Channel}
\date{\today}
\author{ 
Berivan Isik, Francisco Pernice, Tsachy Weissman
}

\usepackage{hyperref}
\usepackage[
backend=biber,
style=alphabetic,
sorting=ynt,
maxbibnames=99
]{biblatex}
\newcommand{\Cunif}{C_{\text{unif}}}
\newcommand{\E}{\mathbb{E}}

\DeclareMathOperator{\BDC}{\mathrm{BDC}}

\newcommand{\DD}{\mathcal{D}}

\newcommand{\N}{\mathbb{N}}
\newcommand{\Z}{\mathbb{Z}}
\newcommand{\1}{\mathbbm{1}}
\newcommand{\pr}{\mathbb{P}}

\newcommand{\FF}{\mathcal{F}}
\newcommand{\GG}{\mathcal{G}}

\newcommand{\R}{\mathbb{R}}

\newcommand{\til}[1]{\widetilde{#1}}
\newcommand{\st}{\;|\;}

\newcommand{\upto}{\uparrow}

\newcommand{\preq}{\mathrel{\stackrel{\makebox[0pt]{\mbox{\normalfont\tiny pr}}}{=}}}

\DeclareMathOperator{\poly}{\mathsf{poly}}

\DeclareMathOperator{\supp}{supp}
\DeclareMathOperator{\sgn}{sgn}

\newcommand{\pnt}[1]{{\small{\textsc{#1}}}}

\theoremstyle{definition}

\newtheorem{thm}{Theorem}
\newtheorem{lem}[thm]{Lemma}
\newtheorem{cor}[thm]{Corollary}
\newtheorem{obs}[thm]{Observation}
\newtheorem{question}[thm]{Question}

\newtheorem{conj}[thm]{Conjecture}
\newtheorem{rem}[thm]{Remark}
\newtheorem{prop}[thm]{Proposition}

\newtheorem{defn}[thm]{Definition}

\numberwithin{thm}{section}

\begin{document}

\maketitle

\begin{abstract}
    We consider the maximum coding rate achievable by uniformly-random codes for the deletion channel. We prove an upper bound that's within 0.1 of the best known lower bounds for all values of the deletion probability $d,$ and much closer for small and large $d.$ We give simulation results which suggest that our upper bound is within 0.05 of the exact value for all $d$, and within $0.01$ for $d>0.75$. Despite our upper bounds, based on simulations, we conjecture that a positive rate is achievable with uniformly-random codes for all deletion probabilities less than 1. Our results imply impossibility results for the (equivalent) problem of compression of i.i.d. sources correlated via the deletion channel, a relevant model for DNA storage.
\end{abstract}

\section{Introduction}
The \emph{binary deletion channel} takes as input a string of bits, and outputs a random subsequence, chosen by deleting each bit of the input independently with a fixed probability $d.$ The deletion channel is an example of a channel with \emph{memory}: the action of the channel on different input bits cannot be decoupled, due to the asynchrony between input and output that is introduced. Memory vastly complicates the analysis of the fundamental limits of communication through the deletion channel, i.e., its \emph{capacity}. Indeed, calculating the capacity of this channel as a closed-form expression has been a notorious open problem in information theory for more than fifty years. For channels without memory, the situation is much better understood. A single-letter variational formula for the capacity has been known since Shannon~\cite{shannon}, and for \emph{symmetric} memoryless channels, the capacity is explicitly given by the mutual information between a uniformly-random (single-letter) input and the corresponding output. This is equivalent to the statement that, for a symmetric memoryless channel, a uniformly-random code achieves capacity with high probability. 

Given the success in understanding the capacity of memoryless channels, early research on the deletion channel focused on extending Shannon's techniques to settings with memory. A major stride in this direction was made by Dobrushin~\cite{dobrushin}, who showed that the capacity of the deletion channel, as well as a large class of other channels with memory, is given by the following natural extension of Shannon's variational formula:
\begin{align}\label{eq:dobrushin}
    \lim_{n\to\infty} \frac{1}{n} \sup_{X^n} I(X^n;Y^n), 
\end{align}
where $I(\cdot\,;\cdot)$ is the mutual information, the supremum is over the distributions of the random variable $X^n$ supported on $\{0,1\}^n$, and $Y^n$ is the output of the binary deletion channel on input $X^n$.\footnote{The dependence on $d$ is suppressed.} Unfortunately, the resulting variational problem is now infinite-dimensional, and a daunting task to compute. And indeed, appealing though it is, Dobrushin's theorem has proved of little help in obtaining an explicit handle on the capacity of the deletion channel.\footnote{A notable exception is the work of Kirsch and Drinea \cite{drinea-info-thry-lb}, who improved on the lower bound in \cite{hard-lower-bound} via Dobrushin's theorem and purely information-theoretic arguments.} 

Faced with the difficulty of the variational problem involved, it is tempting to attempt ``plugging in'' the uniform distribution as the law of $X^n,$ given the success of this strategy for memoryless symmetric channels. This corresponds to evaluating the performance of uniformly-random codes on the deletion channel. Diggavi and Grossglauser \cite{diggavi} attempted this strategy, and showed that uniformly-random codes can achieve a rate of at least $1-h(d),$ for $d\leq 1/2,$ where $h:[0,1]\to[0,1]$ is the binary entropy function, hence re-proving a lower bound on the deletion channel capacity originally due to Gallager \cite{gallager} and Levenshtein \cite{Levenshtein}. However, in the same paper, Diggavi and Grossglauser also obtained a better lower bound by generating a code with a non-uniform input distribution. This was later improved upon by a series of papers \cite{hard-lower-bound, simple-lower-bound-1/9, drinea-info-thry-lb, dalai, Rahmati-Duman-unif-LB}, all of whom made use of non-uniform input processes. In fact, we now know that a uniformly-random code \emph{cannot} attain the capacity of the deletion channel, e.g., by the results of Drmota, Szpankowski and Viswanatha \cite{Drmota}, who showed that the rate achievable by uniformly-random codes is $O((1-d)^{4/3}\log \frac{1}{1-d})$ as $d\to 1,$ while for optimal codes it's known to be $\Theta(1-d)$ \cite{simple-lower-bound-1/9}. 

Nevertheless, in this paper we advocate the view that the rate attainable by uniformly-random codes for the deletion channel is well worth studying in its own right. We offer three justifications, in increasing order of perceived importance by the authors. First, as was pointed out by Mitzenmacher \cite{mit-review}, and is easy to see, uniformly-random codes have a natural maximum a posteriori (MAP) decoding rule: given a received string $Y,$ return the codeword $c$ for which $Y$ appears as a subsequence the most times. As we illustrate in this paper (and as was implicit in \cite{Yanjun}), this fact can be cast in information-theoretic language, leading to a more concrete characterization of the maximum rate achievable by uniformly-random codes than the limit of the mutual information. Second, as is demonstrated by combining the results of this manuscript with the lower bounds in \cite{diggavi, Rahmati-Duman-unif-LB, Yanjun}, by restricting attention to uniformly-random codes, one can get much tighter lower and upper bounds on the maximum achievable rate. This suggests that a characterization of the maximum rate achievable by uniformly-random codes may be a fruitful milestone in the long road towards understanding the deletion channel capacity. Finally, as we describe in Section~\ref{sec:slepian-wolf}, the determination of the maximum rate achievable by uniformly-random codes has a \emph{direct} application, in addition to the indirect application of advancing our understanding of the deletion channel capacity. Namely, after a simple transformation, the maximum rate achievable by uniformly-random codes is equivalent to the minimum rate of compression of a uniform source given side information that is correlated to it via the deletion channel. This connection, which was explored in \cite{Tse-1, Tse-2}, offers a natural generalization of the classical Slepian-Wolf theory of distributed compression of sources with memoryless correlations, to the setting of deletion-based correlations. As such, it has immediate relevance for distributed DNA storage.

\subsection{Slepian-Wolf Compression with Deletion Correlations}\label{sec:slepian-wolf}
In this section, we give a brief self-contained overview of the classical Slepian-Wolf theory of distributed compression of correlated sources \cite{slepian-wolf}. For a more comprehensive introduction to the Slepian-Wolf theory, including proofs, we refer the reader to \cite[Chapter~15.4]{cover-thomas}.

Shannon's theory, among other applications, gives a fundamental limit on the minimum achievable rate of compression of a source $\{X_i\}_{i=1}^\infty:$ for an i.i.d. source, the simplest case, it is given by the entropy $H(X_1).$ It's natural to ask how this fundamental limit may be extended to the case of multiple sources. It turns out that much about the question is already captured by the case of jointly i.i.d. sources $X:=\{X_i\}_{i=1}^\infty$ and $Y:=\{Y_i\}_{i=1}^\infty$, where each pair $(X_i,Y_i)$ is sampled independently from a joint distribution $p(x,y)$, or equivalently, where $X$ is i.i.d., and $Y$ is the output of a memoryless channel on input $X.$ We wish to determine, for any given pair $(R_X,R_Y) \in [0,1]^2$, whether it's possible to compress $X$ at rate $R_X$ and $Y$ at rate $R_Y,$ i.e., to determine the \emph{achievability region} in the unit square. Slepian and Wolf showed that the achievability region is given by
\[
\{(R_X,R_Y): R_X+R_Y\geq H(X_1,Y_1)\}.
\]
The question becomes even more interesting when one imposes the further restriction that the individual sources $X$ and $Y$ should be encoded \emph{independently}, in a distributed fashion, such that the encoder of $X$ doesn't have access to $Y$ and the encoder of $Y$ doesn't have access to $X,$ and then decoded jointly. This is a model, for instance, of a data server with multiple nodes, where communication between nodes is costly and hence each node would like to perform their compression task independently, while taking advantage of the correlations between nodes. Slepian and Wolf's highly unexpected insight was that, despite the distributed nature of the encoding, the achievability region is
\[
\{(R_X,R_Y): R_X\geq H(X_1|Y_1), R_Y\geq H(Y_1|X_1), R_X+R_Y\geq H(X_1,Y_1)\}.
\]
Given these results, it is of interest to explicitly compute the achievability region for natural choices of sources and correlations. A moment of thought reveals that the problem reduces to calculating $H(X_1),H(Y_1)$, $H(X_1|Y_1)$, and $H(Y_1|X_1)$; the achievability region, in the distributed case, is then given by the convex hull of the points
\[
\begin{cases}
(H(X_1|Y_1), H(Y_1)), \\
(H(X_1),H(Y_1|X_1)),\\
(H(X_1|Y_1), 1), \\
(1,H(Y_1|X_1)), \\
(1,1).
\end{cases}
\]
Hence, for sources $X$ and $Y$ with uniform marginals and memoryless correlations, the problem is as easy and as hard as the determination of the maximum rate achievable by uniform codes over the channels with transition probabilities $p(x|y)$ and $p(y|x)$.

Given the subject of this paper, we'd like to extend this from the correlations induced by memoryless channels to those induced by the deletion channel. Ma, Ramchandran and Tse \cite{Tse-1} showed that the Slepian-Wolf results extend almost verbatim (in the ``corner point'' case where one wants to compress $X$ with $Y$ as side information) to the case where $X$ is uniform and $Y$ is the output of the deletion channel on input $X.$ Their proof also easily extends to the whole achievability region, e.g., by the information-spectrum version of the Slepian-Wolf theorem \cite{MIYAKE}, \cite[Section~7.2]{information-spectrum-book}. One simply has to replace all entropies $H(X_1),H(X_1|Y_1),$ etc. above by the corresponding entropy rates
\[
H(X) := \lim_{n\to\infty} \frac{1}{n}H(X_1^n), \qquad H(X|Y) := \lim_{n\to\infty} \frac{1}{n}H(X_1^n|Y_1^n), \qquad \text{etc.}
\]
Hence, in the setting of distributed compression of uniform sources correlated via deletions, calculating the achievability region amounts to computing $H(X|Y)$ and $H(Y|X)$\footnote{This suffices, since for $X$ uniform and $Y$ the corresponding output through the deletion channel, it's trivial that $Y$ is also uniform and hence $H(X)=H(Y)=1.$} Once again, this is easily seen to be equivalent to computing 
\[
 \lim_{n\to\infty} \frac{1}{n}I(X^n;Y^n(X^n)) = 1 - \lim_{n\to\infty}\frac{1}{n}H(X^n|Y^n(X^n)) = 1-d - \lim_{n\to\infty}\frac{1}{n}H(Y^n(X^n)|X^n),
\]
where $X$ is uniform and now $Y^n(X^n)$ is the output of the deletion channel on input $X^n,$ and \emph{not} the first $n$ coordinates of the process $Y$ (hence the entropy rate of $Y^n(X^n)$ is $1-d$ instead of 1). Namely, calculating the minimum rate of Slepian-Wolf coding for uniform sources and deletion correlations is equivalent to evaluating the performance of uniformly-random codes through the deletion channel. 

\subsection{Our Results}\label{sec:our-results}
Our main results are efficiently-computable upper bounds on the maximum rate achievable through the deletion channel by uniformly-random codes. By the discussion in Subsection~\ref{sec:slepian-wolf}, our results imply impossibility results for the distributed compression of deletion-correlated sources. Unlike many of the known lower and upper bounds related to the deletion channel, which often rely on coding-theoretic arguments like side information or specific decoding strategies, our approach is almost entirely based on tools from probability theory.

We now state our results, and then we plot the resulting bounds at the end of this subsection.

\begin{thm}\label{thm:main-thm}
Let $\Cunif:[0,1]\to[0,1]$ denote the maximum rate achievable through the deletion channel by uniformly-random codes. For each $n\in \N$ and $d\in [0,1],$ we have
\begin{align*}
    \Cunif(d) \leq 1-d- h(d)+& \frac{1}{2n} \log(8\pi e \max\{d(1-d) n, 1/6\}) + \sum_{u=0}^n \binom{n}{u} d^u (1-d)^{n-u} \Bigg(\frac{1}{n}\log \Pi_n^u[u] \\
    &+ h(u/n) + \frac{1}{2n}\log(\max\{2\pi u(n-u)/n,1\}) + \frac{\log e}{6n\max\{\min\{u, n-u\},1\}}\Bigg).
\end{align*}
where $h$ is the binary entropy function, and the function $u\mapsto \Pi_n^u[u]$ can be computed in time $O(n^3)$ and space $O(n^2).$ In particular, our upper bound can be evaluated in time $O(n^4)$ and space $O(n^2).$
\end{thm}
We note that the explicit algorithm to compute $\Pi_n^u[u]$ is given in Theorem~\ref{thm:non-uniform-bd}. As a second result, we improve upon the runtime of Theorem~\ref{thm:main-thm}, at the cost of worse error terms.

\begin{thm}\label{thm:Einf-nonunif-ub-efficient}
For all $d\in (0,1),$ $n\in \N,$ we have
\begin{align*}
    \Cunif(d) &\leq 1-d + \inf_{\delta} \Bigg[\frac{1}{n}\log \Pi_n^{(d-\delta)n}[(d-\delta)n]  +\delta + 2^{-nD(d-\delta||d)}+ \frac{1}{n}\log (\max\{2\pi (d-\delta)(1-d + \delta)n,1\}) \\
    &\qquad\qquad\qquad + \frac{1}{n}\log(8\pi e \max\{d(1-d) n, 1/6\})  + \frac{\log e}{3n^2\min\{\max\{d-\delta, 1-d+\delta\},1\}}\Bigg],
\end{align*}
where $u\mapsto \Pi_n^u[u]$ is the same function as in Theorem~\ref{thm:main-thm}, and $D(\cdot||\cdot):[0,1]^2\to \R_+$ is the (base 2) KL-divergence for two Bernoulli distributions. The infimum is over the $\delta>0$ such that $(d-\delta, d+\delta)\subseteq [0,1]$ and $\delta n\in \N.$ In particular, this upper bound can be computed in time $O(n^3)$ and space $O(n^2)$ by choosing a constant-sized grid to take the infimum over.
\end{thm}
The following corollary, which immediately follows from Theorem~\ref{thm:Einf-nonunif-ub-efficient}, simplifies the above expression at a small loss in the error terms.
\begin{cor}
For all $d\in (0,1)$, $\delta>0, n\in\N$ such that $(d-\delta,d+\delta)\subseteq [0,1],\delta n\in \N,$ and $n\geq \min\{d-\delta, 1-d+\delta\}^{-1},$ we have
\[
\Cunif(d) \leq 1-d+ \frac{1}{n}\log \Pi_n^{(d-\delta)n}[(d-\delta)n] +\delta + 2^{-nD(d-\delta||d)} + \frac{1}{n}\log (\pi^2 e^{4/3} n).
\]
\end{cor}

Although Theorem~\ref{thm:Einf-nonunif-ub-efficient} is asymptotically more efficient than Theorem~\ref{thm:main-thm}, our numerical calculations, all of which ran within a few hours on a personal computer, were not large enough for this difference to be noticed. Hence our best numerical upper bounds, plotted below, are given by Theorem~\ref{thm:main-thm}. In addition to our upper bounds and the lower bounds of \cite{diggavi, Rahmati-Duman-unif-LB, Yanjun}, we also plot simulation-based upper and lower bounds on the exact value of $\Cunif$ for all deletion probabilities. These simulation results suggest that our upper bounds are within 0.05 of the true value of $\Cunif$ for all values of $d.$ For a discussion on how accurate these simulation-based bounds can be taken to be, including the exact method and parameters with which they were calculated, we refer the reader to Section~\ref{sec:sim-results}, which can be read immediately after Section~\ref{sec:preliminaries}. Finally, we mention that, although our upper bound comes very close to zero for $d$ large, we conjecture that $\Cunif(d)$ is positive for all $d<1$ (see Section~\ref{sec:conjectures-and-open-problems}).
\begin{figure}[H]
    \centering
    \includegraphics[width=0.7\linewidth]{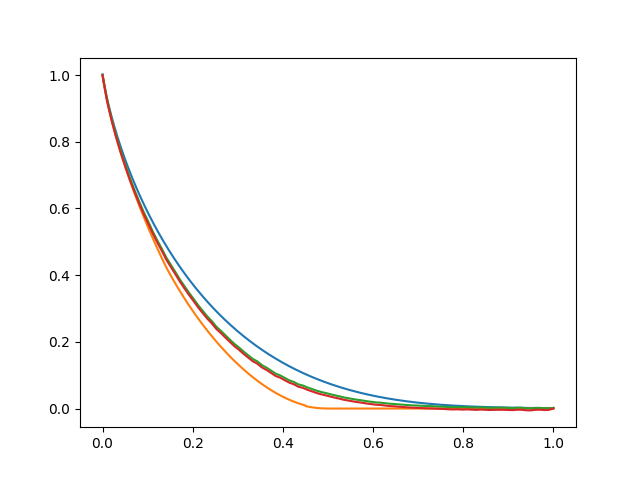}
    \caption{The blue curve is our upper bound on $\Cunif(d)$ from Theorem~\ref{thm:main-thm} for $n=1,000.$ The orange curve is the maximum, for each $d,$ of the lower bounds in \cite{diggavi, Rahmati-Duman-unif-LB, Yanjun}. The green and red curves, which mostly overlap in the figure, are simulation-based (non-proven) upper and lower bounds on $\Cunif$, respectively.}
    \label{fig:fig1}
\end{figure}

\subsection{More Prior Work}
Here we review previously-known explicit bounds on the maximum rate achievable with uniformly-random codes for the deletion channel.

\paragraph{Extremal Regimes.}
Kanoria and Montanari \cite{montanari-1} and Kalai, Mitzenmacher and Sudan \cite{kalai} showed that the capacity of the deletion channel is $1- (1-o(1))h(d)$ as $d\to 0,$ and this bound is attained by uniformly-random codes. Hence we also have $\Cunif(d) = 1- (1-o(1))h(d)$ as $d\to 0.$ In the large-$d$ regime, as we mentioned above, Drmota, Szpankowski and Viswanatha \cite{Drmota} showed that $\Cunif(d) = O((1-d)^{4/3}\log \frac{1}{1-d})$ as $d\to 1.$

\paragraph{Non-Extremal Regime.} As we mentioned above, Diggavi and Grossglauser~\cite{diggavi}, Rahmati and Duman~\cite{Rahmati-Duman-unif-LB} and Han, Ordentlich and Shayevitz \cite{Yanjun} gave lower bounds on $\Cunif$; the best of these bounds, for each $d$, is plotted in Figure~\ref{fig:fig1}. Han, Ordentlich and Shayevitz \cite{Yanjun} and Drmota, Szpankowski and Viswanatha \cite{Drmota} also gave upper bounds on $\Cunif$; we include Figure 3 from \cite{Yanjun} below, which plots these bounds. As far as we know, the upper bound of \cite{Yanjun} is the best known proved specifically for $\Cunif$, instead of the deletion channel capacity. However, it only beats the best known upper bounds on the deletion channel capacity \cite{diggavi-Mit-ub, duman-cap-ub, rahmati-duman-15, Cheraghchi}, which are of course automatically upper bounds on $\Cunif$, for $d$ greater than approximately $0.65.$ By comparison, our bound from Figure~\ref{fig:fig1} beats the best known upper bounds on the deletion channel capacity, as well as the upper bounds on $\Cunif$ of \cite{Drmota, Yanjun}, for all $d \in (0,1).$

\begin{figure}[H]
    \centering
    \includegraphics[width=0.4\linewidth]{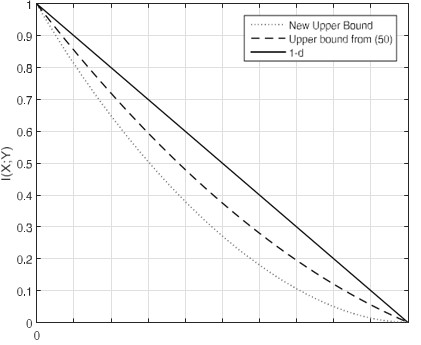}
    \caption{(\cite[Figure~3]{Yanjun}). The upper bounds on $\Cunif$ of \cite{Yanjun} (labelled ``New Upper Bound'') and \cite{Drmota} (labelled ``Upper bound from (50)'').}
    \label{fig:my_label}
\end{figure}

\subsection{Organization}
We organize the rest of the paper as follows. In Section~\ref{sec:preliminaries}, we review background material and derive some simple results that will be crucial for our main proofs. In Section~\ref{sec:upper-bounds}, we prove our main results: in Subsection~\ref{sec:two-subseqs-agree}, we carry out the main calculation at the heart of our proofs, which we apply in Subsection~\ref{sec:unif-bd} to obtain a warmup upper bound. Then in Subsection~\ref{sec:non-unif-bd} we improve upon this to obtain theorems~\ref{thm:main-thm} and \ref{thm:Einf-nonunif-ub-efficient}. In Section~\ref{sec:sim-results}, we discuss our simulation results. Finally, in Section~\ref{sec:conjectures-and-open-problems}, these simulation results lead us to some questions and conjectures, with which we hope to motivate further work.

\section{Preliminaries}\label{sec:preliminaries}
In this section, we establish some basic results that will be needed in Section~\ref{sec:upper-bounds} for our proof of the upper bounds on the maximum rate achievable by uniformly-random codes through the deletion channel. In Subsection~\ref{sec:notation} we fix the notation for the rest of the paper. In Subsection~\ref{sec:initial-obs}, we give a simple proof of a characterization of the maximum rate due to Ma, Ramchandran and Tse \cite{Tse-1}. Along the way, we introduce important quantities that will play a major role throughout the paper. We then point out a simple observation (which was already implicit in \cite{Yanjun}) that leads to our approach for the rest of the paper. Finally in Section~\ref{sec:conv-rate}, following the proof of similar results in \cite{dalai, duman-cap-ub, montanari-1}, we bound the convergence rate of the main relevant quantity (to be introduced below), which will be crucial to us in Section~\ref{sec:upper-bounds}.

\subsection{Notation}\label{sec:notation}

The function $\Cunif:[0,1]\to[0,1]$ denotes the maximum rate achievable by uniformly-random codes for the deletion channel, hereafter also referred to as \emph{the maximum rate} for brevity. For $x\in \{0,1\}^n,$ we let $x_i^j \in \{0,1\}^{j-i+1}$ denote the substring of $x$ starting at index $i$ and ending at $j,$ inclusive; unless otherwise specified, we let $x_i:=x_i^i$ and $x^i:=x_1^i.$ For $x\in \{0,1\}^n,$ we let $|x|=n$ denote the string length and $w(x)$ the Hamming weight. For $x\in \{0,1\}^n,$ $\BDC_d(x)$ denotes the output of the deletion channel with deletion probability $d$ on input $x.$ As we sometimes explicitly point out, we will often use the notation $X$ to refer to a uniformly-distributed random variable in the support $\{0,1\}^n$, and $Y=\BDC_d (X)$. This suppression of the blocklength $n$ in the notation leads to slightly confusing expressions like $\lim_n\frac{1}{n}H(X|Y)$, but it will simplify notation immensely later when we use the sub- and superscripts for parameters other than $n.$ 

For a deletion pattern $D \in \{0,1\}^n$ and an input $x \in \{0,1\}^n$, we use the notation $D(x)$ to refer to the output that's produced by applying deletion pattern $D$ to input $x,$ namely, by deleting the $i$th bit of $x$ if and only if $D_i=1.$ Hence with $X\sim Unif\{0,1\}^n$ and $D \sim Ber(d)^n$ (the $n$-fold product of the Bernoulli measure), we have $Y = \BDC_d (X) = D(X).$ We also use the notation $\supp(D) \subseteq [n]$ to refer to the set of indices $i\in [n]$ such that $D_i=1.$ For a set $A \subseteq [n],$ we denote by $X_A \in \{0,1\}^{|A|}$ the subsequence of $X$ obtained by choosing the indices in $A$ from $X.$ Hence we have e.g. $D(X) = X_{\supp(D)^c}$, where $A^c := [n]\setminus A$ for $A\subseteq [n]$. 

All logs (hence entropies, etc.) in this paper are in base two.

\subsection{Initial Observations}\label{sec:initial-obs}
Ma, Ramchandran and Tse \cite{Tse-1} observed that
\begin{align}\label{eq:tse-formula}
    \lim_{n\to\infty}\frac{1}{n}H(X|Y) =  d + h(d) - \lim_{n\to\infty} \frac{1}{n} H(D|X,Y),
\end{align}
where $h$ is the binary entropy and $D \in \{0,1\}^n$ is the deletion pattern that yielded $Y$ from $X$. This of course immediately implies, by Dobrushin's expression (\ref{eq:dobrushin}) for $\Cunif$, that
\[
\Cunif(d) = 1 -d-h(d) + \lim_{n\to\infty} \frac{1}{n} H(D|X,Y).
\]
Since (\ref{eq:tse-formula}) is the starting point of our work, and the proof in \cite{Tse-1} is more general than what we need, here we give a very simple self-contained proof. Along the way, we obtain explicit rates, which will be useful later.

We have
\begin{align*}
    H(X|Y) &=  H(X,D|D(X)) - H(D|X,Y) \\
    &=H(X_{\supp(D)},D) - H(D|X,Y) \\
    &=H(X_{\supp(D)}) + H(D|w(D)) - H(D|X,Y) \\
    &= H(X_{\supp(D)}) +\E\left[\log\binom{n}{w(D)}\right] - H(D|X,Y),
\end{align*}
so to prove (\ref{eq:tse-formula}), it suffices to show that, as $n\to\infty,$ we have ($a$) that $\frac{1}{n}H(X_{\supp(D)}) \to d$, and ($b$) that $\frac{1}{n}\E[\log\binom{n}{w(D)}] \to h(d)$. For claim $(a)$, let $L := |X_{\supp(D)}| = n- |Y|\sim Binom(n,d).$ We have
\begin{align*}
   dn &= \E_{\ell\sim L}[\ell] \\
   &=  E_{\ell\sim L}[H(X_{\supp(D)}|L=\ell)] \\
   &= H(X_{\supp(D)}|L) \\
   &\leq   H(X_{\supp(D)}) \\
   &= H(X_{\supp(D)}|L) + H(L) \\
    &\leq dn +  \frac{1}{2}\log(2\pi e ( nd(1-d)+1/12)),
\end{align*}
so $dn\leq  H(X_{\supp(D)}) \leq dn + o(n),$ proving ($a$). Above, we have used that the entropy of a binomial with variance $v$ is bounded by $\frac{1}{2}\log(2\pi e ( v+1/12))$ \cite{massey-entropy}. For claim ($b$), note first that since $\log \binom{n}{k}= \log \Gamma(n+1) - \log  \Gamma(k+1) - \log \Gamma (n-k+1)$ and $\Gamma$ is log-convex in $\R_+$, the extended function $[0,n]\ni x\mapsto \log \binom{n}{x}$ is concave. Using this and that $\frac{1}{n}\log\binom{n}{w(D)} \in [0,1]$, and taking $\delta > 0$ such that $(d-\delta,d+\delta)\subseteq [0,1],$ we get
\begin{align*}
    \pr((d-\delta) n \leq w(D) \leq (d+\delta)n)\inf_{d'\in (d-\delta,d+\delta)} \frac{1}{n}\log \binom{n}{d'n}&\leq \frac{1}{n}\E\left[\log\binom{n}{w(D)}\right] \leq \frac{1}{n}\log \binom{n}{dn}.
\end{align*}
By concentration of $w(D)$ and uniform convergence on the left hand side (guaranteed by Dini's theorem\footnote{Dini's theorem states that a monotone limit of continuous functions (to a continuous limit) in a compact set is uniform. That $\frac{1}{n}\log \binom{n}{dn} \upto h(d)$, with $d\mapsto \binom{n}{dn}$ continuously extended via the $\Gamma$ function as above, is a standard fact.}), we get
\begin{align*}
    \inf_{d'\in (d-\delta,d+\delta)} h(d') \leq \liminf_n  \frac{1}{n}\E\left[\log\binom{n}{w(D)}\right] \leq \limsup_n  \frac{1}{n}\E\left[\log\binom{n}{w(D)}\right]\leq h(d),
\end{align*}
and taking $\delta\to 0$, the continuity of $h$ yields claim ($b$), proving (\ref{eq:tse-formula}).

Given the formula (\ref{eq:tse-formula}), to obtain the maximum rate it remains to calculate the term $E_\infty$ defined below:
\begin{align}\label{eq:En-def}
    E_n =E_n(d):= \frac{1}{n} H(D|X,Y),&\qquad E_\infty =E_\infty(d) := \lim_{n\to\infty} E_n(d).
\end{align}
We then have
\begin{align}\label{eq:Cunif-Einf-formula}
    \Cunif(d) = 1-d-h(d) + E_\infty(d).
\end{align}

Our work begins with the following observation, which was essentially made in \cite{Yanjun} in a more general setting.\footnote{In \cite{Yanjun}, the same observation is made in a more general context, but under a further assumption. In the case of the deletion channel, this assumption specializes to the fact that the deletion pattern $D$ has the uniform distribution in some set. Of course, this is not directly true: the law $Ber(d)^n$ of $D$ is not uniform. But it's easy to see, as was pointed out in \cite{Yanjun}, that this point is superficial: $D$ can be taken uniform in the strings of weight $dn$ without changing the asymptotics. We will use this fact later as well.}
\begin{obs}\label{obs:D-unif}
The distribution of $D$ conditioned on $X=x$ and $Y=y$ is uniform in the set\footnote{Below, the font $\pnt{D}$ is used to denote a realization of the random variable $D.$}
\[
\DD(x,y) := \{\pnt{D}\in\{0,1\}^n : \pnt{D}(x) = y\}.
\]
\end{obs}
This is intuitively clear, but since the rest of our work relies on this observation, we give a proof.
\begin{proof}[Proof of Observation~\ref{obs:D-unif}]
Take $\pnt{D} \in \DD(x,y).$ We have
\begin{align*}
    \pr(D = \pnt{D} |X=x,Y=y) &= \frac{\pr(X=x,Y=y|D=\pnt{D}) \pr(D=\pnt{D})}{ \pr(X=x,Y=y)} \\
    &= \frac{\pr(X=x)}{\pr(X=x,Y=y)}  \pr(Y=y|D=\pnt{D},X=x)\pr(D=\pnt{D}) \\
    &= \frac{\pr(X=x)}{\pr(X=x,Y=y)}  d^{|x| - |y|}  (1-d)^{|y|}.
\end{align*}
But the last expression is a function of only $x$ and $y$ (and not of $\pnt{D}$).
\end{proof}

From Observation~\ref{obs:D-unif}, we obtain the following re-writing of $E_n$:
\begin{align}\label{eq:unif-D-formula}
    E_n &= \frac{1}{n}\E [\log|\DD(X,Y)|],
\end{align}
where as before $X\sim Unif\{0,1\}^n$ and $Y=\BDC_d (X).$ In words: draw $X$ uniformly at random and pick a random subsequence from it, where you select each bit in $X$ independently with probability $1-d$; what's the exponent (in expectation) of the number of times your chosen subsequence appears in $X$?

\begin{rem}
As we suggested in the introduction, Observation~\ref{obs:D-unif} implies the simplicity of the MAP decoding rule for uniformly-random codes, and hence it allows us to cast this fact in information-theoretic language via (\ref{eq:unif-D-formula}). Indeed, when decoding uniformly-random codes, one can simply output the codeword $x$ for which the received string $y$ appears as a subsequence the most often, i.e. the $x$ which maximizes $|\DD(x,y)|$, because for $\pnt{D}\in \DD(x,y)$ we have
\begin{align*}
    \pr(X=x|Y=y) &= \frac{\pr(X=x, D=\pnt{D}|Y=y)}{\pr(D=\pnt{D}|Y=y,X=x)} \\
    &=|\DD(x,y)| \frac{\pr(D=\pnt{D})\pr(X=x|D=\pnt{D})\pr(Y=y|X=x,D=\pnt{D}) }{\pr(Y=y)}\\
    &= |\DD(x,y)| d^{|x|-|y|}(1-d)^{|y|}\frac{\pr(X=x)}{\pr(Y=y)}.
\end{align*}
\end{rem}

\subsection{Convergence Rate}\label{sec:conv-rate}
In this subsection, we derive a bound on the speed of convergence of $E_n$ to $E_\infty,$ as defined in (\ref{eq:En-def}). Analogous bounds, with a similar proof, have appeared in the literature \cite{dalai, duman-cap-ub, montanari-1}.

We go back to our original expression for $E_n,$ namely $E_n = \frac{1}{n}H(D|X,Y).$ Let $J$ be the index of $Y$ such that the substring $Y^J$ is the output of the first half of the deletion pattern $D$ on the first half of the input $X,$ i.e. $Y^J = D^{n/2}(X^{n/2}).$ We have
\begin{align*}
    E_n &= \frac{1}{n}H(D|X,Y) \\
    &= \frac{1}{n}H(D,J|X,Y) \\
    &= \frac{1}{n} [H(J|X,Y) + H(D|X,Y,J)] \\
    &= \frac{1}{n} [H(J|X,Y) + 2 H(D^{n/2}|X^{n/2}, Y^J) ] \\
    &= \frac{1}{n} H(J|X,Y) + E_{n/2}.
\end{align*}
Hence 
\begin{align*}
    0\leq E_n - E_{n/2} \leq  \frac{1}{n} H(J). 
\end{align*}
Now, again using that for a binomial $J$ with variance $v$, we have $H(J) \leq \frac{1}{2}\log(2\pi e ( v+1/12))$ \cite{massey-entropy}, we get
\begin{align}\label{eq:En-diff}
    0\leq E_n - E_{n/2} \leq  \frac{1}{2n} \log(\pi e ( d (1-d)n+ 1/6)). 
\end{align}
From this we can deduce a bound on the convergence rate.
\begin{lem} \label{lem:conv-rate}
For each $n\in \N,$ we have
\begin{align*}
    0\leq E_\infty - E_n \leq \frac{1}{2n} \log(8\pi e \max\{d(1-d) n, 1/6\}).
\end{align*}
\end{lem}
\begin{proof}
For $k \in \N$, by (\ref{eq:En-diff}) we have
\begin{align*}
    0 \leq E_{2^k n} - E_n &= \sum_{j=1}^k E_{2^j n} - E_{2^{j-1}n} \\
    &\leq \frac{1}{2n} \sum_{j=1}^k \frac{1}{2^j} \log (2\pi e \max\{d(1-d)2^j n, 1/6\})\\
    &\leq \frac{1}{2n} \left[ \sum_{j=1}^k \frac{j}{2^j} + \log(2\pi e \max\{d(1-d) n, 1/6\}) \sum_{j=1}^k \frac{1}{2^j}\right],
\end{align*}
and taking $k\to\infty$ yields the lemma.
\end{proof}
As a gauge of the speed of convergence guaranteed by Lemma~\ref{lem:conv-rate}, note that at $d=1/2$ (where our bound is loosest), we have shown 
\[
0\leq E_\infty - E_{10,000} \leq 0.00087.
\]

\section{Upper Bounds} \label{sec:upper-bounds}
So far, we have shown convergence rates of $E_n$ to $E_\infty,$ but we haven't given any quantitative information on $E_\infty$, or equivalently, on $\Cunif$. In this section we prove our main result: an upper bound on $E_\infty$, and hence on $\Cunif$. In Section~\ref{sec:two-subseqs-agree}, we give an exact expression, derived through analysis of random walks, for the probability that two independent outputs of the $\BDC$ (not necessarily with the same deletion probability) on a common uniformly random input agree. In Section~\ref{sec:unif-bd}, we use this expression to give an efficient algorithm to compute an upper bound on $E_n$ (and hence, after the appropriate upwards shift from Section~\ref{sec:conv-rate}, on $E_\infty$). At the end of Section~\ref{sec:unif-bd}, we point out a subtle issue that accounts for much of the loss incurred in this first upper bound. Finally in Section~\ref{sec:non-unif-bd}, we correct this issue, to obtain substantially improved upper bounds, which are also efficiently computable.

\subsection{Probability That Two Subsequences Agree}\label{sec:two-subseqs-agree}
For $\alpha,\beta\in [0,1]$, let $D^\alpha\sim Ber(\alpha)^n$ and $D^\beta\sim Ber(\beta)^n$ be independent. The objective of this section is to give an explicit expression for $\pr(D^\alpha(X) = D^\beta(X)),$ where $X\sim Unif\{0,1\}^n.$ In that direction, consider the differences 
\begin{align}\label{eq:differences-def}
    \xi_i &= D_i^\alpha - D_i^\beta \in \{-1,0,1\},
\end{align}
and the random walk $W=\{W_j\}_{j=1}^n$ defined by
\begin{align}\label{eq:walk-def}
    W_j &= \sum_{i=1}^j \xi_i;\qquad W_0:=0.
\end{align}
We then define the filtration
\begin{align}\label{eq:filtration-F}
    \FF_j := \sigma(W_i, D^\alpha_i:i\leq j) =\sigma(D^\beta_i, D^\alpha_i:i\leq j),
\end{align}
and the process $\{P_j\}_{j=1}^n$ adapted to $\{\FF_j\}_{j=1}^n$ by
\begin{subequations}
  \begin{empheq}[left={P_j=\empheqlbrace}]{alignat = 4}
        &1 &\qquad&\text{if }\xi_j \in \{\pm1\} \text{ and } W_{j-1}=0 \text{ or }\xi_j=\sgn(W_{j-1}), \label{eq:a}\\
        &1/2 & &\text{if }\xi_j \in \{\pm1\} \text{ and }W_{j-1}\neq 0 \text{ and }\xi_j\neq\sgn(W_{j-1}),\label{eq:b}\\
        &1  &&\text{if }\xi_j=0\text{ and }W_{j-1}=0,\label{eq:c}\\
        &1& &\text{if }\xi_j=0\text{ and }W_{j-1}\neq 0\text{ and }D_j^\alpha = 1,\label{eq:d}\\
        &1/2 & &\text{if }\xi_j=0\text{ and }W_{j-1}\neq 0\text{ and }D_j^\alpha=0,\label{eq:e}\end{empheq}
    \end{subequations}
where $\sgn :\R\setminus \{0\}\to \{\pm 1\}$ denotes the sign. Next, we need a definition.

\begin{defn}
We say two strings $x\in \{0,1\}^n$ and $y\in \{0,1\}^m$ are \emph{prefix-equal}, and write $x \preq y$ if we have $x^{\min\{n,m\}} = y^{\min\{n,m\}}.$
\end{defn}
We can now state our proposition.
\begin{prop}\label{prop:uniform-bound-rw}
We have $\pr(D^\alpha(X) \preq D^\beta(X) |\FF_n) = \prod_{j=1}^n P_j$.
\end{prop}
Before the proof, we record several important simplifications and consequences. We first define the more restricted filtration $\GG_j := \sigma(W_i:i\leq j) \subseteq \FF_j.$ We note that, conditioned on $\GG_n,$ the random variables $\{P_j\}_{j=1}^n$ are mutually independent, and moreover we have $\E[P_j|\GG_k] = P_j'$ for any $k\geq j,$ where 
\begin{subequations}
  \begin{empheq}[left={P_j'=\empheqlbrace}]{alignat = 4}
        &1 &\qquad&\text{if }\xi_j \in \{\pm1\} \text{ and } W_{j-1}=0 \text{ or }\xi_j=\sgn(W_{j-1}), \label{eq:a'}\\
        &1/2 & &\text{if }\xi_j \in \{\pm1\} \text{ and }W_{j-1}\neq 0 \text{ and }\xi_j\neq\sgn(W_{j-1}),\label{eq:b'}\\
        &1  &&\text{if }\xi_j=0\text{ and }W_{j-1}=0,\label{eq:c'}\\
        &\frac{\alpha \beta + (1-\alpha)(1-\beta)/2}{\alpha\beta + (1-\alpha)(1-\beta)}& &\text{if }\xi_j=0\text{ and }W_{j-1}\neq 0.\label{eq:d'}\end{empheq}
    \end{subequations}
By the conditional mutual independence, we therefore have
\begin{align*}
    \pr(D^\alpha(X) \preq D^\beta(X) |\GG_n) &= \E\left[\prod_{j=1}^n P_j\;\Bigg|\;\GG_n\right] \\
    &= \prod_{j=1}^n P_j',
\end{align*}
and thus
\[
\pr(D^\alpha(X) = D^\beta(X) |\GG_n) = \1_{w(D^\alpha) = w(D^\beta)} \prod_{j=1}^n P_j' = \1_{W_n=0} \prod_{j=1}^n P_j'.
\]
We can simplify this further with a simple observation. A moment of thought reveals that $W_n = 0$ implies that the number of $j\in [n]$ such that $P_j'$ is of type (\ref{eq:a'}) must equal the number of $j$ such that $P_j'$ is of type (\ref{eq:b'}). Hence, defining the new ``symmetrized'' process $\{P_j^{\text{sym}}\}_{j=1}^n$ as
\begin{align*}
    P_j^{\text{sym}} &= 
    \begin{cases}
        1/\sqrt{2} \qquad&\text{if }\xi_j \in \{\pm1\}, \\
        1  &\text{if }\xi_j=0\text{ and }W_j=0,\\
        \frac{\alpha\beta + (1-\alpha)(1-\beta)/2}{\alpha\beta + (1-\alpha)(1-\beta)}  &\text{if }\xi_j=0\text{ and }W_j\neq 0,
    \end{cases}
\end{align*}
we immediately obtain the following corollary.
\begin{cor}\label{cor:prob-calc}
We have $\pr(D^\alpha(X) = D^\beta(X))  =\E[\1_{W_n=0} \prod_{j=1}^n P_j^{\text{sym}}].$
\end{cor}
We finally prove the proposition.
\begin{proof}[Proof of Proposition~\ref{prop:uniform-bound-rw}]
We proceed by induction on $n.$ For $n=1$ we clearly have $1 = \pr(D^\alpha(X)\preq D^\beta(X)|W) = P_1.$ Now suppose we have the desired formula for $n-1$. We consider first the case where $W_{n-1}=0,$ i.e. where $(D^\alpha)^{n-1}(X^{n-1})$ and $(D^\beta)^{n-1}(X^{n-1})$ have the same length. If $\xi_j \in \{\pm 1\}$, we have 
\begin{align*}
    \{(D^\alpha)^{n-1}(X^{n-1})\preq (D^\beta)^{n-1}(X^{n-1})\} &= \{(D^\alpha)^{n-1}(X^{n-1})= (D^\beta)^{n-1}(X^{n-1})\}  \\
    &=   \{(D^\alpha)^{n}(X^n)\preq (D^\beta)^{n}(X^n)\}.
\end{align*}
Moreover if $\xi_j = 0,$ either $X_n$ is deleted by both deletion patterns, or it's deleted by neither and appears as the last bit of both $(D^\alpha)^n(X)$ and $(D^\beta)^n(X)$. Hence in all cases where $W_{n-1}=0$, we have
\begin{align*}
    \pr((D^\alpha)^n(X^n) \preq (D^\beta)^n(X^n)|W^n) &= 1\cdot  \pr((D^\alpha)^{n-1}(X^{n-1}) \preq (D^\beta)^{n-1}(X^{n-1})|W^{n-1}) \\
    &= P_n \cdot \pr((D^\alpha)^{n-1}(X^{n-1}) \preq (D\beta)^{n-1}(X^{n-1})|W^{n-1}),
\end{align*}
and by inductive hypothesis we obtain the desired formula for $n.$ Finally we consider the case where $W_{n-1}\neq 0,$ and assume without loss of generality that $W_{n-1}>0$, i.e. that $(D^\alpha)^{n-1}(X^{n-1})$ is \emph{shorter} than $(D^\beta)^{n-1}(X^{n-1}).$ Note first that $\xi_n=1$ again implies that 
\[
\{(D^\alpha)^{n-1}(X^{n-1})\preq (D^\beta)^{n-1}(X^{n-1})\} = \{(D^\alpha)^{n}(X^n)\preq (D^\beta)^{n}(X^n)\},
\]
and in this case we also appropriately have $P_n=1.$ Next if $\xi_n = -1,$ then $X_n$ is deleted by $\til{D}^n$ and not by $D^n.$ Hence the last bits of $\til{D}^n(X^n)$ and $D^n(X^n)$ come from different indices in $X,$ and they agree with probability $1/2,$ yielding
\[
\pr((D^\alpha)^{n}(X^n)\preq (D^\beta)^{n}(X^n)|W^n) = \frac{1}{2}\pr((D^\alpha)^{n-1}(X^{n-1})\preq (D^\beta)^{n-1}(X^{n-1})|W^{n-1}).
\]
Lastly if $\xi_n=0,$ then if $D^\alpha_n = 1,$ we had $D_n^\alpha=D^\beta_n=1,$ and if $D^\alpha_n = 0,$ we had $D^\alpha_n=D^\beta_n=0.$ In the former case, we again have $\{(D^\alpha)^{n-1}(X^{n-1})\preq (D^\beta)^{n-1}(X^{n-1})\} = \{(D^\alpha)^{n}(X^n)\preq (D^\beta)^{n}(X^n)\},$ while in the latter case we again have $\pr((D^\alpha)^{n}(X^n)\preq (D^\beta)^{n}(X^n)|W^n) = \frac{1}{2}\pr((D^\alpha)^{n-1}(X^{n-1})\preq (D^\beta)^{n-1}(X^{n-1})|W^{n-1}).$ Therefore again
\begin{align*}
    \pr((D^\alpha)^{n}(X^n)\preq (D^\beta)^{n}(X^n)|W^n) = P_n\cdot \pr((D^\alpha)^{n-1}(X^{n-1})\preq (D^\beta)^{n-1}(X^{n-1})|W^{n-1}),
\end{align*}
and by inductive hypothesis we have corroborated the case $W_{n-1}\neq 0,$ proving the proposition.
\end{proof}
\subsection{Warmup: Uniform Measure Upper Bound}\label{sec:unif-bd}
In this section, we give a first (naive) upper bound on $E_n$, as defined in (\ref{eq:En-def}), and hence, by our formula (\ref{eq:Cunif-Einf-formula}), we obtain a first upper bound on $\Cunif.$ At the heart of our proof is an application of the formula derived in the last subsection. 

We begin with an application of Jensen's inequality: 
\begin{align*}
    E_n &= \frac{1}{n}\E_{X,Y} \log |\DD(X,Y)| \\
    &\leq \frac{1}{n}\log \E_{X,Y}|\DD(X,Y)|.
\end{align*}
We wish to calculate $\E_{X,Y}|\DD(X,Y)|$, at fixed $n,$ exactly and efficiently. Letting $\til{D}\sim Unif\{0,1\}^n = Ber(1/2)^n,$ we can re-write $\E_{X,Y}|\DD(X,Y)| = 2^n \pr(D(X)=\til{D}(X))$, where as before $D\sim Ber(d)^n$ and $X\sim Unif\{0,1\}^n$. We combine this with Lemma~\ref{lem:conv-rate} and (\ref{eq:Cunif-Einf-formula}), and record it below.
\begin{lem}[Warmup Upper Bound]\label{lem:unif-bd-in-terms-of-prob}
For each $n\in \N$ and $d\in [0,1]$, we have
\begin{align*}
    \Cunif(d) \leq 2-d-h(d) + \frac{1}{n}\log \pr(D(X) = \til{D}(X)) + \frac{1}{2n}\log(8\pi e \max\{d(1-d)n,1/6\}),
\end{align*}
where $D\sim Ber(d)^n$ and $\til{D}\sim Unif\{0,1\}^n$
\end{lem}
We note that the name of this subsection comes from the fact that the law of $\til{D}$ is chosen uniform in $\{0,1\}^n;$ in the next subsection, we obtain an improvement by choosing a non-uniform measure. For now, we proceed by applying Corollary~\ref{cor:prob-calc} with $\alpha = d, \beta=1/2$, to obtain
\begin{align}\label{eq:uniform-formula}
    \pr(D(X) = \til{D}(X)) &= \E\left[\1_{W_{n} = 0}\prod_{j=1}^n U_j\right]
\end{align}
with
\begin{align*}
    U_j &= 
    \begin{cases}
        1/\sqrt{2} \qquad&\text{if }\xi_j \in \{\pm1\}, \\
        1  &\text{if }\xi_j=0\text{ and }W_j=0,\\
        (1+d)/2 &\text{if }\xi_j=0\text{ and }W_j\neq 0,
    \end{cases}
\end{align*}
and $\{\xi_j\}_{j=1}^n$, $\{W_j\}_{j=1}^n,$ defined as in (\ref{eq:differences-def}), (\ref{eq:walk-def}), respectively. From this we derive an algorithm to compute $\pr(D(X) = \til{D}(X))$.

\begin{thm}\label{thm:unif-bd}
Consider the domain $A := \{(n,k): n\geq -1, |k|\leq n+1\}\subseteq \Z^2$, and its interior $A^\circ = \{(n,k): n\geq 0, |k|\leq n\}$. Define a function $R:A\to \R_+$ by the initial value $R(0,0)=1,$  the recurrence relation
\[
R(n,k) =  R(n-1,k)\cdot  \frac{1}{2}\left(1 -  \frac{1-d}{2}\1_{\{k\neq 0\}}\right) + R(n-1,k-1)\cdot \frac{d}{2\sqrt{2}} + R(n-1,k+1)\cdot \frac{1-d}{2\sqrt{2}}
\]
on $A^\circ\setminus \{(0,0)\},$ and the boundary condition $R\equiv 0$ on $A\setminus A^\circ.$ Then we have $\pr(\til{D}(X) = D(X)) = R(n,0).$ In particular, $\pr(\til{D}(X) = D(X))$ can be computed exactly in time $O(n^2)$ and space $O(n).$
\end{thm}
\begin{proof}
We will prove that, on $A^\circ,$ we have 
\begin{align}\label{eq:rec-ident}
    R(n,k) = \E\left[\1_{W_n = k} \prod_{j=1}^n U_j\right].
\end{align}
The first part of the theorem then immediately follows from (\ref{eq:uniform-formula}). We proceed by induction on $n$. The validity of (\ref{eq:rec-ident}) for the initial condition $n=0$ is immediate: we have $W_0=0$ with probability 1, and an empty product evaluates to 1 by convention. Assume (\ref{eq:rec-ident}) holds for $n-1.$ For a point $(n,k) \in A^\circ,$ consider the filtration $\GG_n = \sigma(W_j : j\leq n)$ defined above. We have
\begin{align*}
    \E\left[\1_{W_n = k} \prod_{j=1}^nU_j\right] &= \E\left[\prod_{j=1}^{n-1}U_j\E\left[\1_{W_n = k} U_n\;\Bigg|\;\GG_{n-1}\right]\right].
\end{align*}
But 
\begin{align*}
    \E[\1_{W_n = k} \cdot U_n \st \GG_{n-1}] &= \E[(\1_{W_{n-1} = k}+\1_{W_{n-1} = k-1} + \1_{W_{n-1} = k+1})\1_{W_n = k} \cdot U_n \st \GG_{n-1}]\\
    &= \1_{W_{n-1} = k} \cdot  \frac{1}{2}\left(\1_{k=0} + \1_{k\neq 0}\frac{1+d}{2}\right) + \1_{W_{n-1}=k-1}\cdot \frac{d}{2}\cdot\frac{1}{\sqrt{2}} + \\
    &\qquad\qquad\qquad\qquad\qquad\qquad\qquad\qquad\qquad\qquad\1_{W_{n-1}=k+1}\cdot \frac{(1-d)}{2}\cdot\frac{1}{\sqrt{2}}, \\
    &=  \1_{W_{n-1} = k} \cdot  \frac{1}{2}\left(1 -  \frac{1-d}{2}\1_{\{k\neq 0\}}\right) + \1_{W_{n-1}=k-1}\cdot \frac{d}{2\sqrt{2}} + \1_{W_{n-1}=k+1}\cdot\frac{1-d}{2\sqrt{2}}.
\end{align*}
Finally for points $(n,k)$ adjacent to the boundary, note that, if $(n-1,\ell) \in A\setminus A^\circ,$ for $\ell\in \{k-1,k,k+1\},$ we have $\1_{W_{n-1}=\ell}\equiv 0.$ Thus, by inductive hypothesis, we have (\ref{eq:rec-ident}), and hence the first part of the theorem.

For the second part of the theorem, the bound on the time complexity is obvious from the first part. To justify the space bound, note that our recursive formula for $R(n,k)$ only involves terms of the form $R(n-1,\ell)$ for some $\ell.$ Hence, we can compute $R$ on the strip $(\{n\}\times \Z) \cap A$ as a function of only $R|_{(\{n-1\}\times \Z)\cap A}$, so at any given time we only need to store the value of $R$ at two adjacent strips $n$-strips, yielding $O(n)$ space.
\end{proof}
Below we plot the upper bound of Lemma~\ref{lem:unif-bd-in-terms-of-prob} for $\Cunif$, and its implied upper bound on $E_\infty$ via (\ref{eq:Cunif-Einf-formula}). We compute these bounds using Theorem~\ref{thm:unif-bd} with $n=10,000.$ We also plot the lower bounds of \cite{diggavi, Rahmati-Duman-unif-LB, Yanjun} and our simulation results. We note that the numerical bound we obtain from this naive upper bound in fact improves upon the numerical bound from Theorem~\ref{thm:main-thm} that we presented in Section~\ref{sec:our-results} for the range of $d\in [0.54, 0.7]$, and by up to 0.008 at $d\approx 0.63$. This is simply because the lower time complexity of this naive bound allows us to easily compute it for larger $n.$
\begin{figure}[H]
\begin{subfigure}{.5\textwidth}
  \centering
  \includegraphics[width=.85\linewidth]{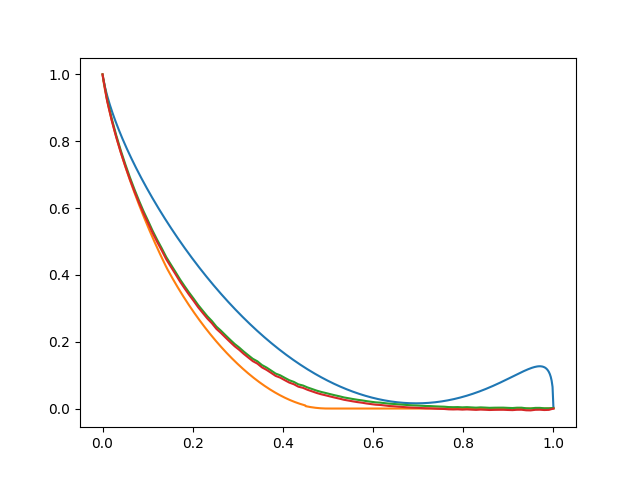}
  \caption{Our upper bound on $\Cunif$ (blue), the corresponding \\ best lower bound (orange), and the simulation-based\\ lower (red) and upper (green) bounds.}
  \label{fig:unif-upper-bound-a}
\end{subfigure}%
\begin{subfigure}{.5\textwidth}
  \centering
  \includegraphics[width=.85\linewidth]{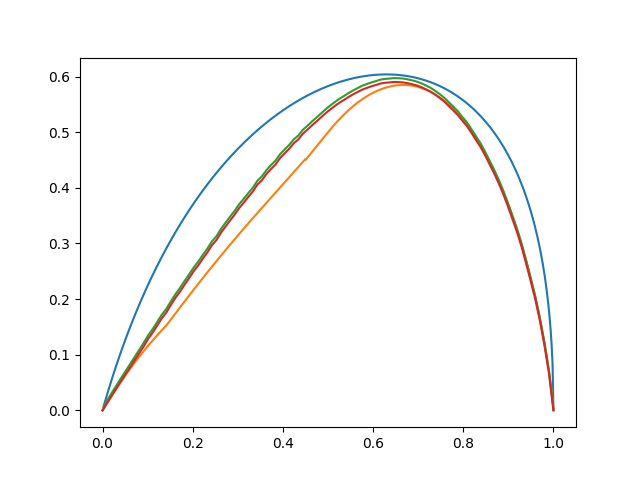}
  \caption{Our upper bound on $E_\infty$ (blue), the best lower bound on $E_\infty$ implied by the lower bounds on $\Cunif$ from subfigure (a) (orange), and the simulation-based lower (red) and upper (green) bounds.}
  \label{fig:unif-upper-bound-b}
\end{subfigure}
\caption{Our upper bounds from Lemma~\ref{lem:unif-bd-in-terms-of-prob} (blue), for $n=10,000$, plotted with the best known corresponding lower bounds (orange), and the simulation-based lower (red) and upper (green) bounds. The orange lower bounds are obtained by taking the maximum of the results in \cite{diggavi, Rahmati-Duman-unif-LB, Yanjun} for each $d.$}
\label{fig:unif-upper-bound}
\end{figure}

The gap in our upper bound can be qualitatively understood. The only loss we incurred came from our application of Jensen's inequality at the beginning of this section. Since we're dealing with an exponential-sized object $|\DD(X,Y)|,$ its expectation can be wildly swayed by events of exponentially small probability. An example of such an event is that the length of the output $L=|Y|$ deviates from its expectation $(1-d)n$ by a constant multiple $\delta$ of $n.$ This happens with probability $O(2^{- \alpha n})$ for some $\alpha=\alpha(\delta)>0$, but if the expected increase 
\begin{align}\label{eq:deviation-increase}
    \log\left(\frac{\E[|\DD(X,Y)| \;|\; |L -(1-d)n| > \delta n]}{\E[|\DD(X,Y)| \;|\; |L- (1-d)n|\leq \delta n]} \right)
\end{align}
is greater than $\alpha,$ then the regime where $|L -(1-d)n| > \delta n$ will impact the value of $\log\E |\DD(X,Y)|.$ A moment of thought reveals that (\ref{eq:deviation-increase}) is small whenever $d$ is at a maximum of $d\mapsto E_n(d)$; hence in this case the loss in our approximation cannot be explained by deviations of $L$ away from its expectation. On the other hand if $E_n(d)$ is increasing at $d,$ then $\E|\DD(X,Y)|$ will be exponentially swayed upwards by the contribution of the event $L \geq (1-d+\delta) n.$ Symmetrically, if $E_n$ is decreasing at $d,$ then $\E|\DD(X,Y)|$ will be exponentially swayed upwards by the contribution of the event $L \leq (1-d-\delta) n.$ And indeed, that our bound is loose when $d\mapsto E_n(d)$ is increasing or decreasing, and close to tight when it's at a maximum, is exactly what we observe in Figure~\ref{fig:unif-upper-bound-b}.

\subsection{Non-Uniform Measure Upper Bound}\label{sec:non-unif-bd}

At the end of the last section we observed qualitatively that our naive bound was loosened substantially by the exponentially-unlikely event that $||Y| - (1-d)n| = \Omega(n).$ Now, what we ultimately want to upper bound is $E_n=\frac{1}{n}\E \log|\DD(X,Y)|,$ which can't possibly be swayed by such events by virtue of $\frac{1}{n}\log|\DD(X,Y)|$ being in $[0,1]$. In this section, we obtain an improved upper bound by ``removing'' these unlikely events from the expectation altogether.

We begin with a definition.
\begin{defn}\label{defn:brackets}
For $X\sim Unif\{0,1\}^n$ and $u \in [n],$ we let $Y[u]$ denote the random variable obtained by deleting \emph{exactly} $u$ uniformly-chosen bits from $X.$ We distinguish this from the notation $Y(d):= \BDC_d (X).$
\end{defn}
To improve our upper bound on $E_\infty$, our approach will be to upper bound 
\[
E_n' = E_n'[u] :=\frac{1}{n}\E\log|\DD(X,Y[u])|.
\]
It's not hard to see that $E_n'[dn]\to E_\infty(d)$ as $n\to\infty$ (this follows, e.g., by \cite[Theorem~1]{dalai}), but we'd like to get a non-asymptotic version of that statement so that we can relate our upper bound on $E_n'$ to $E_\infty$, and hence $\Cunif.$ Unfortunately, our arguments from Subsection~\ref{sec:conv-rate} no longer work because the different bits in the deletion pattern are now dependent. We temporarily disregard this issue, and in Subsection~\ref{sec:upper-bounding-En'}, we focus on upper bounding $E_n'$ using the methods of subsections~\ref{sec:two-subseqs-agree} and \ref{sec:unif-bd}. Then in Subsection~\ref{sec:simple-extension}, we show a simple way to circumvent the issue of convergence of $E_n'$ to $E_\infty,$ by directly upper bounding $E_n$ by an averaging of $E_n'$ values, and then applying the results from Subsection~\ref{sec:conv-rate} to relate $E_\infty$ to $E_n.$ This yields a proof of Theorem~\ref{thm:main-thm}. This method has the advantages of being very simple to prove (once one has an upper bound on $E_n'$), and of having error terms which sharply decay with $n.$ However, it has the disadvantage of making the runtime jump to $O(n^4).$ In Subsection~\ref{sec:controlling-Einf-by-En'}, we instead get a non-asymptotic handle on $E_\infty$ in this setting by relating $E_n(d)$ to $E_n'[(d\pm \delta)n],$ for $\delta$ small. This reduces the runtime to $O(n^3)$, but introduces error terms which decay more slowly. In practice, the best numerical upper bounds are still obtained by the simpler approach of Subsection~\ref{sec:simple-extension}

\subsubsection{Upper bounding $E_n'$}\label{sec:upper-bounding-En'}
In this section, we apply the techniques from subsections \ref{sec:two-subseqs-agree} and \ref{sec:unif-bd} to obtain an upper bound on $E_n'[u],$ for $u\in \{0,1,\dots,n\}$. As in Subsection~\ref{sec:unif-bd}, by Jensen we can write
\begin{align*}
    E_n'[u] &=  \frac{1}{n}\E[\log |\DD(X,Y[u])|] \\
    &\leq  \frac{1}{n}\log \E[|\DD(X,Y[u])|].
\end{align*}
Letting $D, D'\sim_{iid} Ber(u/n)^n$, we have 
\begin{align*}
    \E[|\DD(X,Y[u])|] &= \binom{n}{u}\pr(D(X)=D'(X)|w(D)=w(D') = u)\\
    &= \frac{1}{\binom{n}{u}d^{2u}(1-d)^{2(n-u)}}\pr(D(X)=D'(X),\; w(D)=w(D') = u),
\end{align*}
so defining 
\begin{align}\label{eq:P-def}
    \Pi_n^\ell[u] &:=  \pr(D(X)=D'(X),\; w(D)=w(D') = u), \qquad \qquad D, D'\sim_{iid} Ber(\ell/n)^n,
\end{align}
we have shown
\begin{align*}
    E_n'[u] &\leq \frac{1}{n}\log \Pi_n^u[u] - \frac{1}{n}\log\binom{n}{u} + 2h(u/n).
\end{align*}
Using Lemma~\ref{lem:binom-stirling} from the appendix, we get 
\begin{align}\label{eq:En'-ub}
    E_n'[u] &\leq \frac{1}{n}\log\Pi_n^u[u] +h(u/n) + \frac{1}{2n}\log(\max\{2\pi u(n-u)/n,1\}) +\frac{\log e}{6n
    \max\{\min\{u,n-u\},1\}}.
\end{align}
Our remaining task in this section is to efficiently compute $\Pi_n^u[u].$ From Proposition~\ref{prop:uniform-bound-rw}, setting $\alpha=\beta=\ell/n,$ for $D,D'\sim_{iid}Ber(\ell/n)^n$ we have
\[
\pr(D(X) =  D'(X), w(D) = w(D')=u|\FF_n) = \1_{W_n = 0,\,w(D) =u } \cdot \prod_{j=1}^n P_j.
\]
By the same argument as in Subsection~\ref{sec:unif-bd}, we can replace $P_j$ by the symmetrized process
\begin{align*}
V_j=
    \begin{cases}
    1/\sqrt{2} \qquad &\text{if }\xi_j \in \{\pm1\},  \\
    1  &\text{if }\xi_j=0\text{ and }W_{j-1}=0,\\
    1 &\text{if }\xi_j=0\text{ and }W_{j-1}\neq 0\text{ and }D_j = 1,\\
    1/2  &\text{if }\xi_j=0\text{ and }W_{j-1}\neq 0\text{ and }D_j=0,
    \end{cases}
\end{align*}
from which we obtain the following corollary of Proposition~\ref{prop:uniform-bound-rw}.
\begin{cor}\label{cor:appropriate-weight-probability}
Let $\Pi_n^\ell[u]$ be defined by (\ref{eq:P-def}). We have
\begin{align*}
    \Pi_n^\ell[u] = \E\left[\1_{W_n = 0,\,w(D) =u } \cdot \prod_{j=1}^n V_j\right],
\end{align*}
where $V_j$ is defined as above with $D,D' \sim Ber(\ell/n)^n.$
\end{cor}
Using this, we obtain a recursive relation analogous to Theorem~\ref{thm:unif-bd} to calculate $\Pi_n^u[u]$ efficiently.
\begin{thm}\label{thm:non-uniform-bd}
Consider the domain $B := \{(n,k,u): n\geq -1, |k|\leq n+1, 0\leq u\leq n+1\}\subseteq \Z^3,$ and its interior $B^\circ := \{(n,k,u): n\geq 0, |k|\leq n, 0\leq u\leq n\}.$ For $\ell \in \{0,1,\dots,n\},$ define a function $S^\ell:B\to \R_+$ by the initial value $S^\ell(0,0,0)=1,$ the recurrence relation
\begin{align*}
    S^\ell(n,k,u) &= (\ell/n)^2 S^\ell(n-1, k, u-1) + \frac{(\ell/n)(1-(\ell/n))}{\sqrt{2}} \left(S^\ell(n-1,k-1,u-1)+ S^\ell(n-1,k+1,u)\right) \\
    &\qquad\qquad  + (1-(\ell/n))^2\left(1- \frac{1}{2}\1_{k\neq 0}\right) S^\ell(n-1,k,u).
\end{align*}
on $B^\circ\setminus\{(0,0,0)\}$, and the boundary condition $S^\ell\equiv 0$ on $B\setminus B^\circ.$ Then, if $\Pi_n^\ell[u]$ is defined by (\ref{eq:P-def}), we have $\Pi_n^\ell[u] = S^\ell(n,0,u).$ In particular, $\Pi_n^u[u]$ can be computed in time $O(n^3)$ and space $O(n^2).$
\end{thm}
\begin{proof}
For simplicity, we drop the superscript $\ell$ in the notation and set $d = \ell/n.$ The proof is analogous to that of Theorem~\ref{thm:unif-bd}. We will prove that, on $B^\circ,$ we have
\begin{align}\label{eq:S-formula}
    S(n,k,u) &= \E\left[\1_{W_n = k,\,w(D) = u}\cdot \prod_{j=1}^n V_j\right].
\end{align}
The first part of the theorem then follows from Corollary~\ref{cor:appropriate-weight-probability}. We proceed by induction on $n.$ The validity of (\ref{eq:S-formula}) for the initial condition $n=0$ is immediate: we have $W_0 = 0$ and $w(D^0) = w(\emptyset)=0$ with probability 1, and an empty product evaluates to 1 by convention. Assume (\ref{eq:S-formula}) holds for $n-1.$ For a point $(n,k,u)\in B,$ consider the filtration $\FF_n = \sigma(W_i,D_i:i\leq n)$. We have
\[
\E\left[\1_{W_n = k,\,w(D) = u}\cdot \prod_{j=1}^n V_j\right] = \E\left[\prod_{j=1}^{n-1} V_j \cdot \E\left[\1_{W_n = k,\,w(D) = u}\cdot V_n\;\Bigg|\; \FF_{n-1}\right]\right].
\]
But
\begin{align*}
     \E\left[\1_{W_n = k,\,w(D) = u}\cdot V_n\st \FF_{n-1}\right] &=\E[(\1_{W_{n-1}=k, w(D^{n-1})=u-1}
     +\1_{W_{n-1}=k, w(D^{n-1})=u-1}+\1_{W_{n-1}=k-1, w(D^{n-1})=u-1}\\
     &\qquad\qquad 
     +\1_{W_{n-1}=k+1, w(D^{n-1})=u}
     )\1_{W_n = k,w(D) = u}\cdot V_n\st \FF_{n-1}] \\
     &= \1_{W_{n-1}=k, w(D^{n-1})=u-1} d^2 \cdot 1 
     + \1_{W_{n-1}=k, w(D^{n-1})=u-1} d(1-d)\cdot \frac{1}{\sqrt{2}}\\ 
     &\qquad + \1_{W_{n-1}=k-1, w(D^{n-1})=u-1} (1-d)d\cdot \frac{1}{\sqrt{2}}\\
     &\qquad + \1_{W_{n-1}=k+1, w(D^{n-1})=u} (1-d)^2 \cdot \left(\1_{W_n=0} + \frac{1}{2}\1_{W_n\neq 0}\right) \\
     &= \1_{W_{n-1}=k, w(D^{n-1})=u-1} d^2 \cdot 1 
     + \1_{W_{n-1}=k, w(D^{n-1})=u-1} d(1-d)\cdot \frac{1}{\sqrt{2}}\\ 
     &\qquad + \1_{W_{n-1}=k-1, w(D^{n-1})=u-1} (1-d)d\cdot \frac{1}{\sqrt{2}}\\
     &\qquad + \1_{W_{n-1}=k+1, w(D^{n-1})=u} (1-d)^2 \cdot \left(1- \frac{1}{2}\1_{W_n\neq 0}\right).
\end{align*}
Finally for points $(n,k,u)$ adjacent to the boundary, note that if $(n-1,\ell,t)\in B\setminus B^\circ,$ for $(\ell,t)\in \{(k,u-1), (k-1,u-1), (k+1,u), (k,u)\},$ we have $\1_{W_{n-1}=\ell, w(D^{n-1}) = t}\equiv 0$. Thus, by inductive hypothesis, we have (\ref{eq:S-formula}), and hence the first part of the theorem. 

For the second part of the theorem, again the bound on the time complexity is obvious from the first part. The justification on the space bound is the same as in Theorem~\ref{thm:unif-bd}: we can compute $S$ on the strip $(\{n\}\times \Z^2)\cap B$ as a function of only the values of the previous strip $(\{n-1\}\times \Z^2)\cap B$. Hence all other values can be discarded, yielding $O(n^2)$ space and hence the theorem.
\end{proof}

\subsubsection{Proof of Theorem~\ref{thm:main-thm}}\label{sec:simple-extension}
We let $L:=w(D)$ denote the number of deleted bits, for $D\sim Ber(d)^n.$ By the tower property of conditional expectation, we have
\begin{align}\label{eq:convolution-formula}
    E_n(d) &= \E[\E[\log|\DD(X,Y)| \;| L]] =\E E_n'[L].
\end{align}
Hence an upper bound on $E_n'[u]$ for all $u\in [n]$ implies an upper bound on $E_n(d)$ for all $d,$ and thus on $\Cunif$. We now prove Theorem~\ref{thm:main-thm}.
\begin{proof}[Proof of Theorem~\ref{thm:main-thm}]
Noting that $L\sim Binom(n,d),$ we get $\pr(L=u) = \binom{n}{u}d^u(1-d)^{n-u}.$ Combining this with (\ref{eq:convolution-formula}), (\ref{eq:En'-ub}), (\ref{eq:Cunif-Einf-formula}) and Lemma~\ref{lem:conv-rate} yields
\begin{align*}
    \Cunif(d) \leq 1-d- h(d)+& \frac{1}{2n} \log(8\pi e \max\{d(1-d) n, 1/6\}) + \sum_{u=0}^n \binom{n}{u} d^u (1-d)^{n-u} \Bigg(\frac{1}{n}\log \Pi_n^u[u] \\
    &+ h(u/n) + \frac{1}{2n}\log(\max\{2\pi u(n-u)/n,1\}) + \frac{\log e}{6n\max\{\min\{u, n-u\},1\}}\Bigg).
\end{align*}
The computability of $\Pi_n^u[u]$ in time $O(n^3)$ and space $O(n^2)$ is guaranteed by Theorem~\ref{sec:non-unif-bd}. Since for each $d$ we have to compute $\Pi_n^u[u]$ for $u \in \{0,1,\dots,n\},$ the overall complexity is $O(n^4)$ time and $O(n^2)$ space. This completes the proof.
\end{proof}

We repeat the figure of Subsection~\ref{sec:our-results} with the explicit evaluation of our bound. We also include here the corresponding implied upper bound on $E_\infty.$

\begin{figure}[H]
\begin{subfigure}{.5\textwidth}
  \centering
  \includegraphics[width=.85\linewidth]{upper-lower-and-sims.png}
  \caption{Our upper bound on $\Cunif$ (blue), the corresponding \\ best lower bound (orange), and the simulation-based\\ lower (red) and upper (green) bounds.}
  \label{fig:unif-upper-bound-1}
\end{subfigure}%
\begin{subfigure}{.5\textwidth}
  \centering
  \includegraphics[width=.85\linewidth]{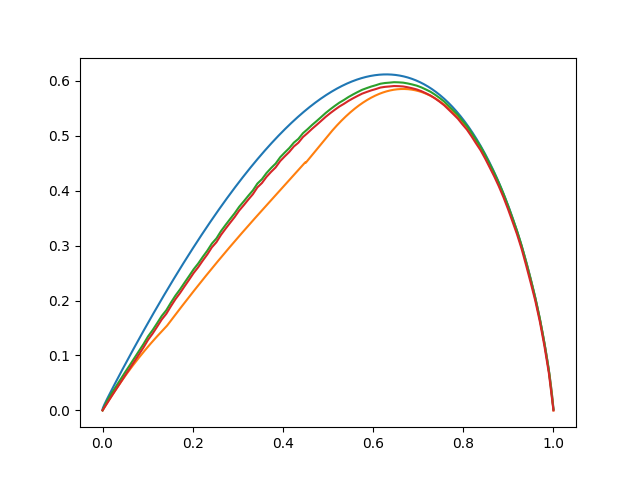}
  \caption{Our upper bound on $E_\infty$ (blue), the best lower bound on $E_\infty$ implied by the lower bounds on $\Cunif$ from subfigure (a) (orange), and the simulation-based lower (red) and upper (green) bounds.}
  \label{fig:unif-upper-bound-2}
\end{subfigure}
\caption{Our upper bounds from Theorem~\ref{thm:main-thm} (blue), for $n=1,000$, plotted with the best known corresponding lower bounds (orange), and the simulation-based lower (red) and upper (green) bounds. The orange lower bounds are obtained by taking the maximum of the results in \cite{diggavi, Rahmati-Duman-unif-LB, Yanjun} for each $d.$}
\label{fig:non-unif-upper-bound}
\end{figure}

\subsubsection{Proof of Theorem~\ref{thm:Einf-nonunif-ub-efficient}}\label{sec:controlling-Einf-by-En'}
In the last subsection, we obtained an improvement over the naive bound of Subsection~\ref{sec:unif-bd}, but the runtime jumped from $O(n^2)$ to $O(n^4).$ In this subsection, we improve this runtime to $O(n^3),$ at the cost of worse error terms.

We  first collect a useful lemma.
\begin{lem}\label{lem:increasing-entropy}
The function $k\mapsto H(X|Y[k])$ is monotone-increasing.
\end{lem}
The (elementary) proof is in the appendix. From this we obtain the following bounds, which are a simple sharpening of \cite[Lemma~4]{dalai}.

\begin{lem}\label{lem:H-concentration}
For all $d,\delta\in [0,1],$ we have
\begin{align*}
    \frac{H(X|Y[(d-\delta)n])}{n} - 2^{-nD(d-\delta||d)}&\leq \frac{H(X|Y(d))}{n} \leq \frac{H(X|Y[(d+\delta) n])}{n} + 2^{-n D(d+\delta || d)},
\end{align*}
where $D(\cdot||\cdot):[0,1]^2\to \R_+$ is the (base 2) KL-divergence for two Bernoulli distributions.
\end{lem}

\begin{proof}
For $K\sim Binom(n,d),$ we can write 
\begin{align*}
    H(X|Y(d)) &= \E_{k\leftarrow K} [H(X|Y[k])].
\end{align*}
But from Lemma~\ref{lem:increasing-entropy} and the fact that $H(X|Y[k]) / n\in [0,1]$ for all $k \in [n]$, we immediately deduce
\begin{align*}
    (1-\pr(K \leq (d-\delta)n) )\frac{H(X|Y[(d-\delta)n])}{n} \leq \frac{H(X|Y(d))}{n} \leq \frac{H(X|Y[(d+\delta) n])}{n} + \pr(K\geq (d+\delta)n).
\end{align*}
The lemma then follows by a standard large-deviations bound on the binomial distribution, and again the fact that $H(X|Y[(d-\delta)n]) / n \leq 1.$
\end{proof}

Now we perform the derivation analogous to (\ref{eq:tse-formula}) for the case of $H(X|Y[\alpha n]),$ for $\alpha \in [0,1].$ 
Noting that the deletion pattern, denoted $D',$ is now distributed uniformly in the strings of Hamming weight $\alpha n,$ we have
\begin{align*}
    H(X|Y[\alpha n]) &= H(X,D'|D'(X)) - H(D'|X,Y[\alpha n]) \\
    &= H(X_{\supp(D')}) + H(D') - H(D'|X,Y[\alpha n])\\
    &= dn + \log \binom{n}{\alpha n} - n\cdot E_n'[\alpha n].
\end{align*}
In the last step we have used the fact that, conditioned on $X=x$ and $Y[\alpha n] = y,$ $D'$ is again uniform on $\DD(x,y)$, which follows by a trivial modification of our proof of Observation~\ref{obs:D-unif}. Putting this together with Lemma~\ref{lem:H-concentration} and our bounds obtained in the derivation of (\ref{eq:tse-formula}), we get
\begin{align*}
    d-\delta + \frac{1}{n}\log \binom{n}{(d-\delta)n} - E_n'[(d-\delta)n]-2^{-nD(d-\delta||d)}  &\leq d + \frac{1}{2n}\log (2\pi e (nd(1-d)+1/12)) + \frac{1}{n}\log \binom{n}{dn} - E_n(d).
\end{align*}
Reordering and combining with Lemma~\ref{lem:conv-rate} yields the following: for all $d\in (0,1)$ and $n\in \N,$ we have 
\begin{align}\label{eq:E-infty-En'-bd-eq}
    E_\infty(d) &\leq \inf_{\delta > 0} E'_n[(d-\delta)n] +\delta+ 2^{-nD(d-\delta||d)} +\frac{1}{n}\left[\log \binom{n}{dn} - \log \binom{n}{(d-\delta)n} + \log(8\pi e \max\{d(1-d) n, 1/6\})\right],
\end{align}
where the infimum is over the $\delta>0$ such that $d-\delta \geq 0.$ We can now prove Theorem~\ref{thm:Einf-nonunif-ub-efficient}.
\begin{proof}[Proof of Theorem~\ref{thm:Einf-nonunif-ub-efficient}]
Putting together (\ref{eq:E-infty-En'-bd-eq}) and (\ref{eq:En'-ub}), and using that $\frac{1}{n}\log\binom{n}{dn} \leq h(d)$ for all $n,$ we get
\begin{align*}
    E_\infty(d) &\leq \inf_{\delta>0} \frac{1}{n} \log \Pi_n^{(d-\delta)n}[(d-\delta)n] + h(d) +\delta+ 2^{-nD(d-\delta||d)}\\
    &\qquad \qquad +2\left[h(d-\delta) - \frac{1}{n}\log\binom{n}{(d-\delta)n}\right] + \frac{1}{n}\log(8\pi e \max\{d(1-d) n, 1/6\}).
\end{align*}
Combining this with Lemma~\ref{lem:binom-stirling} in the appendix and (\ref{eq:Cunif-Einf-formula}) yields 
\begin{align*}
    \Cunif(d) &\leq 1-d + \inf_{\delta} \Bigg[\frac{1}{n}\log \Pi_n^{(d-\delta)n}[(d-\delta)n]  +\delta + 2^{-nD(d-\delta||d)}+ \frac{1}{n}\log (\max\{2\pi (d-\delta)(1-d + \delta)n,1\}) \\
    &\qquad\qquad\qquad + \frac{1}{n}\log(8\pi e \max\{d(1-d) n, 1/6\})  + \frac{\log e}{3n^2\max\{\min\{d-\delta, 1-d+\delta\},1\}}\Bigg],
\end{align*}
as desired. The final part of the theorem, regarding runtime, is true by Theorem~\ref{thm:non-uniform-bd}, as in the proof of Theorem~\ref{thm:main-thm}.
\end{proof}

\section{Simulating $E_n$}\label{sec:sim-results}
Given $x$ and $y,$ the quantity $|\DD(x,y)|$ can be calculated in time $O(|x||y|)$ and auxiliary space $O(|y|)$ by a standard dynamic programming algorithm (see, e.g., \cite{subsequence-counts-book, mit-review}). We define an array $a$ of length equal to $|y|.$ We iterate an index $i=0,\dots,|x|-1.$ At each fixed value of $i,$ the idea is to update our array $a$ such that, for each $j=0,\dots,|y|-1,$ the value $a[j]$ counts the number of times that $y^{j+1}$ appears as a subsequence in $x^{i+1}.$ For fixed $i$, this can be computed as a function of $x_{i+1}, y_{j+1}$, and the previous values of $a[j-1], a[j]$ (from iteration $i-1$): if $x_{i+1}= y_{j+1}$, we can now extend all our subsequences $y^j$ from the previous iteration by appending $x_i$ to them; hence $a[j] += a[j-1]$. Otherwise if $x_{i+1}\neq  y_{j+1}$, we leave the value $a[j]$ unchanged in this iteration. The formal algorithm is below.
\begin{algorithm}[H]
\caption{A standard algorithm to compute the number of times that $y$ appears as a subsequence of $x.$}\label{alg:cap}
\begin{algorithmic}
\State \text{\bf Input:} $x\in \{0,1\}^n,y\in \{0,1\}^m, n\geq m.$
\State $a \gets$ Array of length $m.$
\For{$i \in \{0,1,\dots, n-1\}$}
\For{$j \in \{m-1,m-2,\dots, 1\}$}
\If{$x_{i+1} = y_{j+1}$}
    $a[j] \gets a[j] + a[j-1]$
\EndIf
\EndFor
\If{$x_{i+1} = y_1$}
    $a[0] \gets a[0] + 1$
\EndIf
\EndFor  \\
\Return $a[m-1]$
\end{algorithmic}
\end{algorithm}
\begin{rem}
Note two subtleties. First, we treat the case of $j=0$ separately: $a[0]$ counts the number of times that the bit $y_1$ appears in $x$, so here a match $x_{i+1}=y_1$ results in an increment of $a[0]$ by 1. Second, we are iterating over $y$ backwards: this is so that our updates of $a[j]$ at a fixed outer iteration $i$ don't interfere with each other.
\end{rem}

Using this algorithm to calculate $|\DD(X,Y)|$ for several independently-sampled $(X,Y)$ pairs, where $X\sim Unif\{0,1\}^n$ and $Y=\BDC_d(X),$ taking log, dividing by $n$, and averaging, we can simulate the value of $E_n = \frac{1}{n}\E \log |\DD(X,Y)|$. By Lemma~\ref{lem:conv-rate}, our simulated value of $E_n$ gives a simulation-based lower bound on $E_\infty$, and after an upward shift by $\frac{1}{2n}\log (8\pi e \max\{d(1-d)n,1/6\})$, we obtain a simulation-based upper bound. The simulation-based lower and upper bounds of figures \ref{fig:fig1}, \ref{fig:unif-upper-bound} and \ref{fig:non-unif-upper-bound} were obtained by this method, all with $n=1,000$ and by averaging $1,000$ independent samples.

It's natural to ask how accurate these simulation results should be expected to be. A complete answer to this question would require a good understanding of the concentration properties of the random variable $\frac{1}{n}\log |\DD(X,Y)|$, which seems like a very challenging problem. Simulation results suggest that it concentrates very strongly (see conjectures~\ref{conj:weak-concentration} and \ref{conj:strong-concentration} in the next section). However, we note that \emph{something} can be said just from the fact that $\frac{1}{n}\log |\DD(X,Y)| \in [0,1]$ almost surely. Namely, letting $Z_n^t$ be the empirical average of $t$ independent samples equal in distribution to $\frac{1}{n}\log |\DD(X,Y)|,$ by Hoeffding's inequality we have
\[
\pr(|Z_n^t - E_n| \geq \gamma) \leq 2e^{-2\gamma^2 t}
\]
for any $\gamma > 0.$ This implies, for instance, a $99\%$ confidence interval around our simulated lower and upper bounds of width equal to $2\cdot 0.0514$, since $2e^{-2(0.0514)^21000}\approx 0.01$. However, this is of course very weak, and in fact already implied by our upper bounds and the lower bounds of \cite{diggavi, Rahmati-Duman-unif-LB, Yanjun}. This confidence interval could be easily tightened by running larger simulations, but we didn't attempt to do this since the empirically-observed concentration is so strong.

\section{Conjectures and Open Problems}\label{sec:conjectures-and-open-problems}
We conclude by collecting some conjectures and open problems with which we hope to motivate further work. We begin with the concentration properties of $\frac{1}{n}\log |\DD(X,Y)|,$ which we already discussed in the last section. An initial conjecture, which may not be so hard to prove, is the following.

\begin{conj}[Concentration conjecture, weak form.]\label{conj:weak-concentration}
$\frac{1}{n}\log|\DD(X,Y)|\to E_\infty$ almost surely.
\end{conj}
The following strengthening implies Conjecture~\ref{conj:weak-concentration} by Lemma~\ref{lem:conv-rate} and the Borel-Cantelli lemma.
\begin{conj}[Concentration conjecture, strong form.]\label{conj:strong-concentration} 
For any $\gamma > 0,$ we have
\[
\pr\left(\left|\frac{1}{n} \log|\DD(X,Y)| - E_n\right| \geq \gamma \right) \leq e^{-\Omega(n/\poly(\log (n)))}.
\]
\end{conj}

Finally, we turn to what we consider to be a fascinating question.
\begin{question}
Determine the value $d^*_{\text{unif}} \in [\frac{1}{2},1]$ when $\Cunif$ first becomes equal to zero. Formally, determine
\[
d^*_{\text{unif}} := \sup\{d\in [0,1]:\Cunif(d)>0\}.
\]
\end{question}
Based on simulation results, we put forth the following.
\begin{conj}
$d^*_{\text{unif}} = 1.$
\end{conj}
We remark that, as far as we know, the state of the art is that $\frac{1}{2}\leq d^*_{\text{unif}}\leq 1,$ i.e., we don't have a lower bound that beats 1/2.  Given the recent progress on the analogous question for the case of the adversarial deletion channel capacity \cite{Ray-threshold}, perhaps related techniques could be applied here as well.

\section{Acknowledgement}
The second author is grateful to Amir Dembo, Mary Wootters, and Ray Li for helpful discussions.

\printbibliography

@article{simple-lower-bound-1/9,  
author={Mitzenmacher, M. and Drinea, E.},  
journal={IEEE Transactions on Information Theory},   title={A Simple Lower Bound for the Capacity of the Deletion Channel},  
year={2006}
}

@article{hard-lower-bound,  
author={Drinea, Eleni and Mitzenmacher, Michael},  journal={IEEE Transactions on Information Theory},   title={Improved Lower Bounds for the Capacity of i.i.d. Deletion and Duplication Channels},  
year={2007}
}

@article{drinea-info-thry-lb,
  author={Kirsch, Adam and Drinea, Eleni},
  journal={IEEE Transactions on Information Theory}, 
  title={Directly Lower Bounding the Information Capacity for Channels With I.I.D. Deletions and Duplications}, 
  year={2010},
  volume={56},
  number={1},
  pages={86-102},
  doi={10.1109/TIT.2009.2034883}
}

@INPROCEEDINGS{diggavi-Mit-ub,  author={Diggavi, Suhas and Mitzenmacher, Michael and Pfister, Henry D.},  booktitle={2007 IEEE International Symposium on Information Theory},   title={Capacity Upper Bounds for the Deletion Channel},   year={2007},  volume={},  number={},  pages={1716-1720},  doi={10.1109/ISIT.2007.4557469}}

@article{diggavi,
author = {Diggavi, Suhas and Grossglauser, Matthias},
year = {2001},
month = {12},
title = {On Transmission Over Deletion Channels},
journal = {Proc. Annu. Allerton Conf. Commun. Control Comput}
}

@inproceedings{massey-entropy,
  title={On the entropy of integer-valued random variables},
  author={Massey, James L},
  booktitle={Procs. of Int. Workshop on Information Theory, Beijing, China},
  pages={C1},
  year={1988}
}

@inproceedings{Tse-1,
author = {Ma, Nan and Ramchandran, Kannan and Tse, David},
year = {2011},
month = {09},
pages = {583 - 587},
title = {Efficient file synchronization: A distributed source coding approach},
doi = {10.1109/ISIT.2011.6034196}
}

@INPROCEEDINGS{kalai,  author={Kalai, Adam and Mitzenmacher, Michael and Sudan, Madhu},  booktitle={2010 IEEE International Symposium on Information Theory},   title={Tight asymptotic bounds for the deletion channel with small deletion probabilities},   year={2010},  volume={},  number={},  pages={997-1001},  doi={10.1109/ISIT.2010.5513746}}

@ARTICLE{slepian-wolf,  author={Slepian, D. and Wolf, J.},  journal={IEEE Transactions on Information Theory},   title={Noiseless coding of correlated information sources},   year={1973},  volume={19},  number={4},  pages={471-480},  doi={10.1109/TIT.1973.1055037}}

@INPROCEEDINGS{Tse-2,  author={Ma, Nan and Ramchandran, Kannan and Tse, David},  booktitle={2012 IEEE International Symposium on Information Theory Proceedings},   title={A compression algorithm using mis-aligned side-information},   year={2012},  volume={},  number={},  pages={16-20},  doi={10.1109/ISIT.2012.6283542}}

@book{cover-thomas,
author = {Cover, Thomas M. and Thomas, Joy A.},
title = {Elements of Information Theory (Wiley Series in Telecommunications and Signal Processing)},
year = {2006},
isbn = {0471241954},
publisher = {Wiley-Interscience},
address = {USA}
}

@ARTICLE{rahmati-duman-15, 
author={Rahmati, Mojtaba and Duman, Tolga M.},  
journal={IEEE Transactions on Information Theory},   
title={Upper Bounds on the Capacity of Deletion Channels Using Channel Fragmentation},   
year={2015}
}

@inproceedings{gallager,
  title={SEQUENTIAL DECODING FOR BINARY CHANNELS WITH NOISE AND SYNCHRONIZATION ERRORS},
  author={Robert G. Gallager},
  year={1961}
}

@article{dobrushin,
  title={Shannon’s theorems for channels with synchronization errors},
  author={Dobrushin, Roland L'vovich},
  journal={Problemy Peredachi Informatsii},
  year={1967}
}

@article{Levenshtein,
  title={Binary codes capable of correcting deletions, insertions, and reversals},
  author={Vladimir I. Levenshtein},
  journal={Soviet physics. Doklady},
  year={1965},
  volume={10},
  pages={707-710}
}

@book{information-spectrum-book,
  title={Information-spectrum methods in information theory},
  author={Koga, H and others},
  volume={50},
  year={2013},
  publisher={Springer Science \& Business Media}
}

@article{MIYAKE,
  title={Coding theorems on correlated general sources},
  author={MIYAKE, Shigeki and KANAYA, Fumio},
  journal={IEICE Transactions on Fundamentals of Electronics, Communications and Computer Sciences},
  volume={78},
  number={9},
  pages={1063--1070},
  year={1995},
  publisher={The Institute of Electronics, Information and Communication Engineers}
}

@article{duman-cap-ub,
author = {Fertonani, Dario and Duman, Tolga M.},
title = {Novel Bounds on the Capacity of the Binary Deletion Channel},
year = {2010},
issue_date = {June 2010},
publisher = {IEEE Press},
volume = {56},
number = {6},
issn = {0018-9448},
url = {https://doi.org/10.1109/TIT.2010.2046210},
doi = {10.1109/TIT.2010.2046210},
journal = {IEEE Trans. Inf. Theor.},
month = {jun},
pages = {2753–2765},
numpages = {13}
}

@INPROCEEDINGS{Ray-threshold,  author={Guruswami, Venkatesan and He, Xiaoyu and Li, Ray},  booktitle={2021 IEEE 62nd Annual Symposium on Foundations of Computer Science (FOCS)},   title={The zero-rate threshold for adversarial bit-deletions is less than 1/2},   year={2022},  volume={},  number={},  pages={727-738},  doi={10.1109/FOCS52979.2021.00076}}

@book{subsequence-counts-book,
author = {Lothaire, M.},
title = {Applied Combinatorics on Words (Encyclopedia of Mathematics and Its Applications)},
year = {2005},
isbn = {0521848024},
publisher = {Cambridge University Press},
address = {USA}
}

@article{Cheraghchi,
author = {Cheraghchi, Mahdi},
title = {Capacity Upper Bounds for Deletion-Type Channels},
year = {2019},
journal = {J. ACM}
}

@ARTICLE{shannon, 
author={Shannon, C. E.},  
journal={The Bell System Technical Journal},   
title={A mathematical theory of communication},  
year={1948}
}

@article{Rahmati-Duman-unif-LB,
author = {Rahmati, Mojtaba and Duman, Tolga},
year = {2013},
month = {09},
pages = {5534-5546},
title = {Bounds on the Capacity of Random Insertion and Deletion-Additive Noise Channels},
volume = {59},
journal = {Information Theory, IEEE Transactions on},
doi = {10.1109/TIT.2013.2262019}
}

@ARTICLE{Yanjun,  author={Han, Yanjun and Ordentlich, Or and Shayevitz, Ofer},  journal={IEEE Transactions on Information Theory},   title={Mutual Information Bounds via Adjacency Events},   year={2016},  volume={62},  number={11},  pages={6068-6080},  doi={10.1109/TIT.2016.2609390}}

@INPROCEEDINGS{Drmota,  author={Drmota, Michael and Szpankowski, Wojciech and Viswanathan, Krishnamurthy},  booktitle={2012 IEEE International Symposium on Information Theory Proceedings},   title={Mutual information for a deletion channel},   year={2012},  volume={},  number={},  pages={2561-2565},  doi={10.1109/ISIT.2012.6283980}}

@INPROCEEDINGS{dalai,
author={Dalai, Marco},  
booktitle={2011 IEEE International Symposium on Information Theory Proceedings},  
title={A new bound on the capacity of the binary deletion channel with high deletion probabilities},   year={2011}
}

@article{mit-review,
author = {Michael Mitzenmacher},
title = {A survey of results for deletion channels and related synchronization channels},
journal = {Probability Surveys},
publisher = {Institute of Mathematical Statistics and Bernoulli Society},
pages = {1 -- 33},
year = {2009},
}

@INPROCEEDINGS{montanari-1,  author={Kanoria, Yashodhan and Montanari, Andrea},  booktitle={2010 IEEE International Symposium on Information Theory},   title={On the deletion channel with small deletion probability},   year={2010},  volume={},  number={},  pages={1002-1006},  doi={10.1109/ISIT.2010.5513745}}
\appendix

\section{Omitted Proofs}
\begin{lem}[Lemma~\ref{lem:increasing-entropy} in the main text]
The function $k\mapsto H(X|Y[k])$ is monotone-increasing.
\end{lem}
\begin{proof}
Take $0\leq i\leq j\leq n.$ We wish to prove that $H(X|Y[i])\leq  H(X|Y[j]).$ Since $X\to Y[i] \to Y[j]$ forms a Markov chain, by the data processing inequality we have $I(X;Y[i])\geq I(X;Y[j]).$ But since $H(X) = n,$ this implies $H(X|Y[i])\leq H(X|Y[j]),$ as desired.
\end{proof}

\begin{lem}\label{lem:binom-stirling}
For $u \in \{0,1,\dots,n\}$ and $n\geq 1$, we have
\[
\frac{1}{n}\log \binom{n}{u} \geq h(u/n) - \frac{1}{2n}\log (\max\{2\pi u(n-u)/n,1\}) - \frac{\log e}{6n\max\{\min\{u, (n-u)\}, 1\}}.
\]
\end{lem}
\begin{proof}
The lemma follows from Stirling's approximation 
\[
 m\log m - m\log e + \frac{1}{2}\log(2\pi m) + \frac{1}{12 m+1} \log e\leq \log m! \leq m\log m - m\log e + \frac{1}{2}\log(2\pi m) + \frac{1}{12 m} \log e
\]
applied to each term in the expression $\log\binom{n}{u} = \log n! - \log (u)! - \log (n-u)!,$ and the simple inequality
\[
\frac{1}{12 n - 1} - \frac{1}{12u} - \frac{1}{12(n-u)} \geq \frac{-1}{6n\min\{u, n-u\}}.
\]
We skip the details of the computation. The maxima with 1 are to fix the cases where $u = 0$ or $n.$
\end{proof}

\end{document}